\documentclass[aps,prl,twocolumn,notitlepage,superscriptaddress,showpacs,10pt,longbibliography,floatfix]{revtex4-1} 

\usepackage{amsmath}    
\usepackage{amsfonts}
\usepackage{bm}
\usepackage{amssymb}
\usepackage{graphicx}   
\usepackage[usenames,dvipsnames]{xcolor}      
\usepackage{comment}
\usepackage{siunitx}
\usepackage{enumitem}
\usepackage{upgreek}
\usepackage{color}

\usepackage{soul}
\usepackage[normalem]{ulem}

\mathchardef\mhyphen="2D

\newcommand{\spec}{\mbox{\rm spec}}

\begin{document}

\title{Local markers for crystalline topology}

\author{Alexander Cerjan}
\email[]{awcerja@sandia.gov}
\affiliation{Center for Integrated Nanotechnologies, Sandia National Laboratories, Albuquerque, New Mexico 87185, USA}

\author{Terry A.\ Loring}
\affiliation{Department of Mathematics and Statistics, University of New Mexico, Albuquerque, New Mexico 87131, USA}

\author{Hermann Schulz-Baldes}
\affiliation{FAU Erlangen-N{\" u}rnberg, Department Mathematik, Cauerstr.\ 11, D-91058 Erlangen, Germany}

\date{\today}

\begin{abstract}
Over the last few years, crystalline topology has been used in photonic crystals to realize edge- and corner-localized states that enhance light-matter interactions for potential device applications. However, the band-theoretic approaches currently used to classify bulk topological crystalline phases cannot predict the existence, localization, or spectral isolation of any resulting boundary-localized modes. While interfaces between materials in different crystalline phases must have topological states at some energy, these states need not appear within the band gap, and thus may not be useful for applications. Here, we derive a class of local markers for identifying material topology due to crystalline symmetries, as well as a corresponding measure of topological protection. As our real-space-based approach is inherently local, it immediately reveals the existence and robustness of topological boundary-localized states, yielding a predictive framework for designing topological crystalline heterostructures. Beyond enabling the optimization of device geometries, we anticipate that our framework will also provide a route forward to deriving local markers for other classes of topology that are reliant upon spatial symmetries.
\end{abstract}

\maketitle

The discovery of crystalline-symmetry protected topological phases, such as obstructed atomic limits \cite{bradlyn_topological_2017}, fra\-gile topology \cite{po_fragile_2018,wang_higher_fragile_2019,de_paz_tqc_2019,song_twisted_2020}, and higher-order topology \cite{benalcazar2017quad,benalcazar2017quadPRB,song_densuremath-2-dimensional_2017,xie_higher-order_2021}, has played a prominent role in the development of artificial topological materials. 
Indeed, one of the primary features of such materials is that their geometry can be carefully tailored during fabrication, allowing for exquisite control over a system's spatial symmetries \cite{schulz_topological_2021,maldovan_sound_2013,xue_topological_2022}.
In photonic crystals, the edge- and corner-localized modes that can appear at the interfaces between structures in different topological crystalline phases have been used to realize a wide variety of useful phenomena, such as lasers \cite{zeng_electrically_2020,yang_spin-momentum-locked_2020,shao_high-performance_2020,kim_multipolar_2020,gong_topological_2020,dikopoltsev_topological_2021}, single photon routing \cite{barik_topological_2018,barik_chiral_2020,parappurath_direct_2020,arora_direct_2021,hauff_chiral_2022}, and structures for enhancing harmonic generation \cite{smirnova_third-harmonic_2019,shalaev_robust_2019,ota_photonic_2019,kruk_nonlinear_2021}. Crystalline topology can also be observed in acoustic systems \cite{serra2018observation,ni_observation_2019,xue_acoustic_2019,xue_realization_2019,zhang_second-order_2019,zhang_dimensional_2019,ni_demonstration_2020,xue_observation_2020,zhang_symmetry-protected_2020,peri_experimental_2020}, where it can protect Fano resonances \cite{zangeneh-nejad_fano_2019} and enable robust analog signal processing \cite{zangeneh-nejad_topological_2019}.

However, the existing theoretical framework for identifying crystalline topology poses a substantial challenge for the design of many types of artificial materials seeking to leverage these phases' boundary-localized phenomena. At present, this classification framework is rooted in band theory, and diagnoses a system's topology through the calculation of elementary band representations \cite{bradlyn_topological_2017,cano_building_2018,de_paz_tqc_2019} and symmetry indicators \cite{kruthoff_topological_2017,po_symmetry-based_2017,benalcazar_quantization_2019,watanabe_space_2018,mondragon-shem_robust_2019,velury_topological_2021,gupta_wannier-function_2022,christensen_location_2022,ghorashi_prevalence_2023}, or Wilson loops over a system's Brillouin zone \cite{zak_berrys_1989,alexandradinata_wilson_loop_2014,Wang_2019}.
Yet, the interface between gapped materials that are in different crystalline-symmetry protected topological phases is not guaranteed to exhibit a localized state at the center of their common band gap, or even within this gap at all \cite{vaidya_topo_phases_2023}.
Instead, after designing such a topological heterostructure, the existence, localization, and spectral isolation of any boundary states must be confirmed \textcolor{black}{through additional analysis, such as} large-volume simulations of the interface. Although it is possible to combine crystalline symmetries with chiral or particle-hole symmetry to protect the boundary-localized states' frequencies to be at mid-gap \cite{jung_exact_2021,vaidya_topo_phases_2023}, many artificial materials, including photonic crystals, cannot realize these additional symmetries. Thus, the most salient properties of many artificial crystalline topological materials for enhancing interactions cannot, in general, be predicted or protected by known band-theoretic approaches.

Here, we introduce a class of local markers for identifying the topology of materials due to their crystalline symmetries. These markers are applicable to both first-order and higher-order topology, \textcolor{black}{and changes in them} directly reveal a system's topological boundary-localized states. 
\textcolor{black}{Moreover, associated with every such marker is a local measure of topological protection, providing a quantitative assessment for the robustness of each boundary-localized state.}
We show how this framework can be applied to realistic photonic crystals to identify topological corner-localized states that are nearly degenerate with surrounding edge modes. \textcolor{black}{Furthermore, by calculating the local measure of protection for disordered versions of this system, we demonstrate that a topological state's robustness can be independent of its spectral separation from the bulk bands, in contrast to the standard assumption that such a state's robustness is determined by this spectral separation.}
Looking forward, our framework may both enable the prediction of devices predicated on this class of topology while inherently accounting for finite system size effects, and yield insights for deriving local markers for other classes of topology that are reliant upon spatial symmetries, such as those found in moir\'{e} systems.

\begin{figure*}
    \centering
    \includegraphics{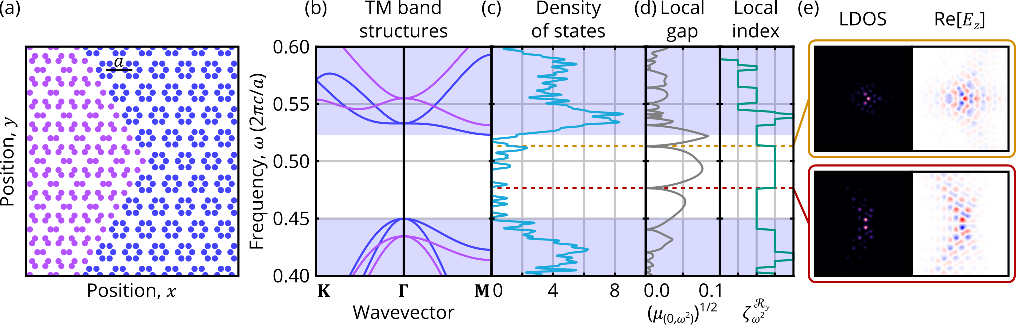}
    \caption{(a) Diagram of a 2D photonic structure with a $120^\circ$ corner between crystals formed from expanded (purple) and contracted (blue) hexagonal clusters, bounded by a perfect electric conductor. The high dielectric rods $\varepsilon = 11.7$ embedded in air have radius $r = a/9$ and are offset from being a honeycomb lattice by $\pm 0.06a$, where $a$ is the lattice constant. (b) Bulk TM band structure for the expanded (purple) and contracted (blue) photonic crystals. (c) Density of states for the finite system in (a). (d) Local gap $(\mu_{(0,\omega^2)})^{1/2}$ in units of $2\pi c/a$ and local index $\zeta_{\omega^2}^{\mathcal{R}_y}$ calculated using $\kappa = 0.01(2\pi c)^2/a^3$. \textcolor{black}{Note, $(\mu_{(0,\omega^2)})^{1/2}$ has units of frequency, enabling direct comparison against the system's band structure.} In (b)-(d) the shaded regions demarcate those frequencies where bulk states exist. (e) Local density of states (LDOS) at the frequency of the local gap closing and real part of the $E_z$ field for the nearest mode of the system. Orange corresponds to $\omega = 0.515(2\pi c/a)$, and red to $\omega = 0.480(2\pi c/a)$.}
    \label{fig:corner}
\end{figure*}

To provide a specific system that exemplifies the difficulties faced in developing artificial topological heterostructures based on crystalline symmetries, we consider a 2D photonic structure consisting of a triangular lattice whose unit cells are decorated with expanded or contracted hexagons of high-dielectric rods \cite{wu_scheme_2015} that are arranged to form an interface with a $120^\circ$ corner between the two decoration choices [Fig.\ \ref{fig:corner}(a)]. The transverse magnetic (TM) modes of these two different decoration choices have been previously shown to be in different topological crystalline phases \cite{wu_scheme_2015,yves_crystalline_2017,gorlach_far-field_2018,smirnova_third-harmonic_2019,kruk_nonlinear_2021}, and can exhibit topological corner-localized states within their common band gap [Fig.\ \ref{fig:corner}]. However, as photonic crystals do not generally possess chiral or particle-hole symmetry, these corner-localized states do not appear at the center of the shared gap (without fine tuning). Thus, even if the topological distinction between the two domains is protected by the bulk band gap, this gap does not protect the localization of the corner states, which could become degenerate with the bulk bands for weaker perturbation strengths than are necessary for a bulk topological phase transition (usually resulting in the delocalization of the corner states \cite{benalcazar_bound_2020,cerjan_observation_2020,cerjan_local_2022}). 

A framework designed to identify the topological interface-localized states stemming from the system's crystalline symmetries requires two components: an invariant that distinguishes topological phases, and an argument showing that shifts in the invariant guarantee the appearance of boundary-localized states. Here, we build such a framework by starting with the spectral localizer \cite{loringPseuspectra,fulga_aperiodic_2016,LoringSchuBa_odd,LoringSchuBa_even,doll2021skew,schulz2021spectral,schulz2022invariants,schulz-baldes_spectral_2023}, which is known to be connected to topology arising from local discrete symmetries (i.e., the Altland-Zirnbauer classes \cite{schnyder2008,kitaev2009,ryu2010topological}). The spectral localizer is a composite operator that combines the eigenvalue problems of a \textcolor{black}{finite} system's Hamiltonian $(H-E{\bf 1})\boldsymbol{\upphi} = 0$ and position operators $(X-x{\bf 1})\boldsymbol{\upphi} = 0$ using a Clifford representation. For a system with a single relevant position operator, the spectral localizer can be written as
\begin{align} 
    &L_{(x,E)}(X,H) = (H-E{\bf 1}) \otimes \sigma_x+\kappa(X-x{\bf 1}) \otimes \sigma_y \label{eq:Lfull} \\ 
        & =\begin{pmatrix}
        0 & \!\!\!\!(H-E{\bf 1})-i\kappa(X-x{\bf 1})\\
        (H-E{\bf 1})+i\kappa(X-x{\bf 1}) & \!\!\!\! 0
        \end{pmatrix}, \notag
\end{align} 
where the Pauli matrices $\sigma_x$ and $\sigma_y$ are used as the Clifford representation. Here, $\kappa>0$ is a tuning coefficient to ensure consistent units and comparable contributions of all summands, and ${\bf 1}$ is the identity. \textcolor{black}{The approximate scale of $\kappa$ is set by the bulk band gap $g$ and the length of the finite system $l$ in the relevant dimension, $\kappa \approx 2g/l$, see Supplemental Material Sec.~SIII \cite{SI}.} In cases where the matrix arguments are implied by their context, they will be omitted, e.g., $L_{(x,E)}=L_{(x,E)}(X,H)$.

Unlike standard eigenvalue equations, where the eigenvalues are determined by their respective operators,  the position $x$ and energy $E$ are inputs in the spectral localizer, and its spectrum quantifies whether the system exhibits a state approximately localized at $(x,E)$ \cite{cerjan_quadratic_2022}, or how large of a system perturbation $\delta H$ is needed to obtain such a state. In particular, if the minimum distance over all of the eigenvalues of $L_{(x,E)}$ to $0$, 
\begin{equation}
\mu_{(x,E)}(X,H) = \min(|\spec[{L}_{(x,E)}(X,H)]|), \label{eq:mu}
\end{equation}
is small relative to $\Vert [H,\kappa X] \Vert$, such an approximately localized state exists \cite{cerjan_quadratic_2022}. Here, $\spec[L]$ denotes the spectrum of $L$. Conversely, if $\mu_{(x,E)}$ is large, a perturbation with norm $\Vert \delta H \Vert \gtrapprox \mu_{(x,E)}(X,H)$ is required for such a state to be found, i.e., for $\mu_{(x,E)}(X,H + \delta H) = 0$. As such, $\mu_{(x,E)}$ can be heuristically understood as a ``local band gap''.

As crystalline symmetries are independent of a system's local discrete symmetries, a crystalline invariant should be applicable to any system regardless of the presence or absence of such discrete symmetries. Thus, a local crystalline topological marker should be given by the signature of an invertible Hermitian matrix, i.e.\ its number of positive eigenvalues minus its number of negative ones. This is analogous to how, for example, the $0$th Chern number of a 0D system is given by the partitioning of its Hamiltonian's eigenvalues about a chosen band gap; the relevant invertible Hermitian matrix here is $H - E_{\textrm{g}}{\bf 1}$, where $E_{\textrm{g}}$ is the band gap's central energy \cite{kitaev2006anyons}. (In contrast, topology originating from phenomena like parity switches are not related to a matrix's signature, but can only manifest in the presence of specific local discrete symmetries.) However, even though $L_{(x,E)}$ is Hermitian, its block off-diagonal structure guarantees that its eigenvalues are always symmetric about $0$ for any choice of $x$ and $E$.

Instead, we seek to remove the duplication in $\spec[{L}_{(x,E)}]$ using the system's crystalline symmetry, and then define an invariant based on this reduced spectrum. In particular, a local crystalline topological marker can be constructed from $L_{(x,E)}$ if the system has a unitary crystalline symmetry $\mathcal{S}$ that satisfies $\mathcal{S}^2 = {\bf 1}$,  $H \mathcal{S} = \mathcal{S} H$, and $X \mathcal{S} = - \mathcal{S} X$. Multiplying the off-diagonal blocks in Eq.\ \eqref{eq:Lfull} by $\mathcal{S}$ at $x=0$ yields the {\it symmetry-reduced spectral localizer} 
\begin{equation}
\tilde{L}_{E}^{\mathcal{S}} (X,H) =(H-E{\bf 1}+i\kappa X) \mathcal{S},
\end{equation}
(A related operator can be constructed for 1D chiral symmetric systems \cite{cerjan_local_2022,Cheng2023}.)
Remarkably, $\tilde{L}_{E}^{\mathcal{S}} =(\tilde{L}_{E}^{\mathcal{S}})^\dagger$ is Hermitian due to the above symmetry relations. Even though $\tilde{L}_{E}^{\mathcal{S}}$ is only built from a single block of $L_{(0,E)}$, it contains all of the essential spectral information in $L_{(0,E)}$, as 
\begin{equation}
\lambda \in \spec[\tilde{L}_{E}^{\mathcal{S}}(X,H) ]\; \Longrightarrow\; \pm \lambda \in \spec[L_{(0,E)}(X,H)], \label{eq:specLink}
\end{equation}
see Supplemental Material Sec.~SI \cite{SI}.
Hence, $\tilde{L}_{E}^{\mathcal{S}} $ has only ``half of the eigenvalues'' of $L_{(0,E)}$, and these eigenvalues need not lie symmetrically around $0$. Thus, a local crystalline topological marker 
can be constructed as
\begin{equation}
\zeta_{E}^{\mathcal{S}}(X,H) = \tfrac{1}{2}\, \textrm{sig}[\tilde{L}_{E}^{\mathcal{S}}(X,H)], \label{eq:zeta}
\end{equation}
where $\textrm{sig}[\tilde{L}_{E}^{\mathcal{S}}]$ is the matrix's signature. For a system with an even/odd number of states, $\zeta_{E}^{\mathcal{S}}$ is integer/half-integer, but the changes in $\zeta_{E}^{\mathcal{S}}$ are always integer-valued (and define the spectral flow, which provides a rigorous generalization to the thermodynamic limit, see Supplemental Material Sec.~SI \cite{SI}). Note that $\zeta_{E}^{\mathcal{S}}$ is only defined when $\tilde{L}_{E}^{\mathcal{S}} $ possesses a spectral gap about $0$, that is, $\mu_{(0,E)} \ne 0$. Moreover, one can prove that pairs $(X,H)$ describing finite systems with the same $\zeta_{E}^{\mathcal{S}}$ can be path-connected to each other while preserving $\mathcal{S}$ and maintaining $\mu_{(0,E)} > 0$, while this is impossible for systems with different $\zeta_{E}^{\mathcal{S}}$, see Supplemental Material Sec.~SII \cite{SI}. 

The local marker $\zeta_{E}^{\mathcal{S}}$ is an indicator for topological boundary states. Specifically, $\zeta_{E}^{\mathcal{S}}$ can only change its value at some energy $E_{\textrm{c}}$ if $0 \in \spec[ \tilde{L}_{E_{\textrm{c}}}^{\mathcal{S}} ]$ so that one of the eigenvalues can switch its sign. But, due to Eq.\ \eqref{eq:specLink}, this requires that $\mu_{(0,E_{\textrm{c}})} = 0$. In turn, a vanishing local gap guarantees the existence of an eigenvalue of $H$ near $E_{\textrm{c}}$ whose corresponding eigenstate is centered at $x=0$ \cite{cerjan_quadratic_2022,SI}. If the local gap closing occurs within a bulk band gap, it must correspond to a boundary-localized state. Furthermore, $\mu_{(0,E)} \ne 0$ provides a measure of topological protection, as a perturbation must close the local gap for the topology to change. Thus, altogether, $\zeta_{E}^{\mathcal{S}}$ both distinguishes crystalline topological phases with respect to $\mathcal{S}$ and changes in its value guarantee that the system possesses a topological state. 

To demonstrate that $\zeta_{E}^{\mathcal{S}}$ is a useful invariant for predicting the behavior of artificial topological materials, we apply the spectral localizer framework to the heterostructure considered in Fig.\ \ref{fig:corner}(a). We start with the second-order differential equation form of Maxwell's time-harmonic equations for TM modes 
\begin{equation}
    \nabla^2 E_z(\mathbf{x}) = - \omega^2 \varepsilon(\mathbf{x}) E_z(\mathbf{x}), \label{eq:diffMax}
\end{equation}
in which $E_z(\mathbf{x})$ is the $z$-component of the electromagnetic field, $\omega$ is the frequency, $\varepsilon(\mathbf{x})>0$ is the spatially dependent dielectric distribution, and the magnetic permeability is assumed to be the identity. 
Using standard finite-difference methods to approximate the Laplacian, we obtain the pair of finite matrices $\nabla^2 \rightarrow W$ and $M$, such that $M$ can be diagonal with $[M]_{j,j} = \varepsilon(\mathbf{x}_j)$, where $\mathbf{x}_j = (x_j,y_j)$ is the $j$th vertex in the discretization. Thus, Eq.\ \eqref{eq:diffMax} can be written as the Hermitian eigenvalue problem
\begin{equation}
    (- M^{-1/2} W M^{-1/2} - \omega^2 {\bf 1})\boldsymbol{\uppsi} = 0, \label{eq:matMax}
\end{equation}
where $\boldsymbol{\uppsi} = M^{1/2}E_z$. Note that if $[M,\mathcal{S}] = 0$, then $[M^{-1/2},\mathcal{S}] = 0$. The discretization of the system also defines its position operators, which can also be chosen to be diagonal; for the 2D system considered in Fig.\ \ref{fig:corner}, $[X]_{j,j} = x_j$ and $[Y]_{j,j} = y_j$. Overall, this formulation of Maxwell's equations and the subsequent choice of discretization are chosen to preserve a dielectric distribution's crystalline symmetries.

The heterostructure with a $120^\circ$ corner in its interface that is considered in Fig.\ \ref{fig:corner} possesses a reflection symmetry $\mathcal{R}_y$ about the $y=0$ axis. Thus, as $H \mathcal{R}_y = \mathcal{R}_y H$, $Y \mathcal{R}_y = - \mathcal{R}_y Y$, and $\mathcal{R}_y^2 = {\bf 1}$, this symmetry can be used to define a $\mathcal{R}_y$-symmetrized spectral localizer and associated local marker as 
\begin{equation}
\tilde{L}^{\mathcal{R}_y}_{\omega^2}=(H-\omega^2 {\bf 1}+i\kappa Y) \mathcal{R}_y,
\quad
\zeta_{\omega^2}^{\mathcal{R}_y} = \tfrac{1}{2}\,\textrm{sig}[\tilde{L}^{\mathcal{R}_y}_{\omega^2}], 
    \label{eq:zetaR}
\end{equation}
where $H = -M^{-1/2} W M^{-1/2}$, and the local gap at $y=0$ is given by $\mu_{(0,\omega^2)} = \min(|\spec[\tilde{L}^{\mathcal{R}_y}_{\omega^2}]|)$. Although the system in Fig.\ \ref{fig:corner}(a) is 2D, $\tilde{L}^{\mathcal{R}_y}_{\omega^2}$ and the associated local index and gap only use one of its two position operators (similar to real-space formulations of other weak invariants \cite{schulz2021spectral,schulz2022invariants}). This effectively forces $\zeta_{\omega^2}^{\mathcal{R}_y}$ and $\mu_{(0,\omega^2)}$ to focus on the system's reflection center. 

As can be seen in Fig.\ \ref{fig:corner}(d),(e), the corner heterostructure exhibits two topological corner-localized states within its bulk band gap. Although these states are difficult to uniquely identify in the system's density of states (DOS) due to the surrounding, nearly degenerate edge-localized modes [Fig.\ \ref{fig:corner}(c) and Fig.\ S2], the corner states can be immediately identified using the system's local gap, as they are energetically close to the local gap closing within the heterostructure's bulk band gap and do not come in reflection symmetric pairs (see Supplemental Material Sec.~SVII \cite{SI}). Moreover, at both of these closings the local topological index changes, proving that both corner states are topological with respect to $\mathcal{R}_y$. Finally, the large local gaps on either side of these corner states indicate that they are both robust against fabrication imperfections. \textcolor{black}{Quantitatively similar results appear for a range of $\kappa$, see Supplemental Material Sec.~SIII \cite{SI}.} Note that although there are many local index changes within the spectral extent of the bulk bands, the tiny local gaps within these regions indicate that these topological phases have vanishing protection.


\textcolor{black}{In contrast to the assumption that crystalline topological states closer to the center of the common bulk band gap are better protected against disorder, the local gap [Fig.\ \ref{fig:corner}(d)] reveals that the higher-frequency corner-localized state in this system has more topological protection than the lower-frequency corner state despite being closer-in-frequency to the bulk bands. In general, a perturbation $\delta H$ with strength $\Vert \delta H \Vert \gtrapprox \mu_{(0,\omega^2)}$ is necessary to change the system's local topology, but this criterion is not sufficient for photonic systems---an arbitrary $H + \delta H$ cannot generally be decomposed into a physically meaningful combination of a local permittivity and Laplacian, per Eq.\ \eqref{eq:matMax}. Instead, the increased protection of the higher-frequency state can be seen by finding the dielectric defect strength necessary to annihilate each of the corner states.
In Fig.\ \ref{fig:defects} we consider two different perturbations that respect $\mathcal{R}_y$, each tailored to affect one corner mode by changing the permittivity of the rod(s) the state has its maximum support on [Fig.~\ref{fig:corner}(e)]. For the lower-frequency corner state's perturbation [Fig.~\ref{fig:defects}(a),(c)], a change in the permittivity of $\delta \varepsilon = 1.17$ is needed to annihilate the topological state by combining it with a state from the lower-frequency bulk bands. In comparison, the necessary perturbation to annihilate the higher-frequency corner state is $\delta \varepsilon = -4.28$  [Fig.~\ref{fig:defects}(b),(d)], despite a similar overlap of the corner-localized state and the perturbation (see Supplemental Material Sec.~SIV \cite{SI}). These results provide evidence that the local gap yields an experimentally relevant hierarchy of protection for a photonic system's topological states.}

\textcolor{black}{The local crystalline topological marker $\zeta_{E}^{\mathcal{S}}$ is also applicable to first-order topology. Section SV in the Supplemental Material provides an example of classifying edge states using this framework \cite{SI}.}

\begin{figure}
    \centering
    \includegraphics{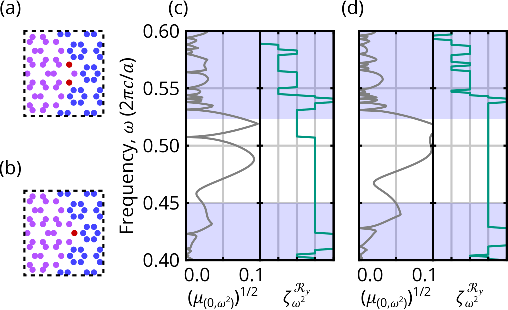}
    \caption{\textcolor{black}{(a),(b) Zoomed in diagram of the perturbed rods (red) in the photonic crystal corner heterostructure from Fig.\ \ref{fig:corner}(a) tailored to the lower-frequency (a) and higher-frequency (b) corner states. (c) Local gap $(\mu_{(0,\omega^2)})^{(1/2)}$ in units of $2\pi c/a$ and local index $\zeta_{\omega^2}^{\mathcal{R}_y}$ calculated using $\kappa = 0.01(2\pi c)^2/a^3$ for the lower-frequency perturbation with $\delta \varepsilon = \varepsilon_{\textrm{red}} - \varepsilon = 1.17$. (d) Similar to (c), except using the higher-frequency perturbation with $\delta \varepsilon = -4.28$. The shaded regions in (c),(d) demarcate those frequencies within the bulk bands.}}
    \label{fig:defects}
\end{figure}


Having proved that $\zeta_{E}^{\mathcal{S}}$ is a useful local marker indicating the existence of topological boundary states and that $\mu_{(0,E)}$ is its associated measure of protection, we now provide a physical interpretation for changes in its value. 
Consider the eigenvalue $l_E$ of $\tilde{L}_{E}^{\mathcal{S}}$ closest to zero and its corresponding eigenvector $\boldsymbol{\upphi}_E$, and note the square of the symmetry-reduced spectral localizer,
\begin{equation}
(\tilde{L}_{E}^{\mathcal{S}})^2 = (H-E{\bf 1})^2 - i\kappa [H,X] + \kappa^2 X^2.
\end{equation}
For systems with only local couplings, $\Vert [H,X] \Vert$ is proportional to the lattice constant times the system's energy scale, and as such is of order $1$ in the system's natural units. Thus, as generally $\kappa$ is small (which guarantees the robustness of the spectral localizer to different choices of $\kappa$ \cite{LoringSchuBa_odd}), to leading order one finds that
\begin{equation}
    l^2_E \boldsymbol{\upphi}_E \approx (H-E{\bf 1})^2 \boldsymbol{\upphi}_E.
\end{equation}
Now, let $E_{\textrm{c}}$ be an energy where the local gap closes $\mu_{(0,E_{\textrm{c}})} = 0$, such that $l_{E_{\textrm{c}}} = 0$. If the spectrum of $H$ is non-degenerate (possibly through the addition of a small amount of symmetry-preserving disorder), one finds $\boldsymbol{\upphi}_{E_{\textrm{c}}} \approx \boldsymbol{\uppsi}_{\textrm{b}}$ and $E_{\textrm{c}} \approx E_{\textrm{b}}$, where $H\boldsymbol{\uppsi}_{\textrm{b}} = E_{\textrm{b}} \boldsymbol{\uppsi}_{\textrm{b}}$. Thus, in the vicinity of the local gap closing where there is only a single relevant eigenstate of the Hamiltonian, one finds that
\begin{align}
    l_E \approx s_{\textrm{b}}((E_{\textrm{b}} - E) + i \kappa \boldsymbol{\uppsi}_{\textrm{b}} ^\dagger X \boldsymbol{\uppsi}_{\textrm{b}}),
\end{align}
where $\mathcal{S} \boldsymbol{\uppsi}_{\textrm{b}} = s_{\textrm{b}} \boldsymbol{\uppsi}_{\textrm{b}}$ with $s_{\textrm{b}} = \pm 1$. But, as $X$ anti-commutes with $\mathcal{S}$, $\boldsymbol{\uppsi}_{\textrm{b}}^\dagger X \boldsymbol{\uppsi}_{\textrm{b}}=0$. Altogether, for $E \approx E_{\textrm{c}}$,
\begin{equation}
    l_E \approx -s_{\textrm{b}}(E - E_{\textrm{c}}), \label{eq:lE}
\end{equation}
i.e., the eigenvalue of $\tilde{L}_{E}^{\mathcal{S}}$ closest to zero is linear near a local gap closing, and the change $\pm 1$ of $\zeta^{\mathcal{S}}_E$ across $E_{\textrm{c}}$ is opposite to the symmetry eigenvalue of the Hamiltonian's corresponding topological state. (See Supplemental Material Sec.~SI \cite{SI} for a more detailed derivation of Eq.\ \eqref{eq:lE}.) The prediction of Eq.\ \eqref{eq:lE} is realized in the system from Fig.\ \ref{fig:corner}; simulations show that for the corner-localized state near $\omega = 0.480(2\pi c/a)$, the index $\zeta_{\omega^2}^{\mathcal{R}_y}$ increases by $+1$ (for increasing $\omega$) when the local gap closes, and the corresponding eigenstate of the system is odd ($s_{\textrm{b}} = -1$) with respect to $\mathcal{R}_y$. The opposite behavior is observed for the corner-localized state near $\omega = 0.515(2\pi c/a)$, with $\zeta_{\omega^2}^{\mathcal{R}_y}$ decreasing as the corner-localized eigenstate is even with respect to $\mathcal{R}_y$. Thus, $\zeta_{E}^{\mathcal{S}}$ is identifying atomic limits with different numbers of states that are either even or odd with respect to $\mathcal{S}$.

In conclusion, we have introduced a class of local crystalline topological markers $\zeta^{\mathcal{S}}_{E}$ and their associated measure of topological protection rooted in the spectral localizer. \textcolor{black}{Unlike traditional theories of crystalline topology that only yield $\mathbb{Z}_N$ invariants \cite{benalcazar_quantization_2019,wieder2018axion,peterson_fractional_2020,yuan_filling_2021,vaidya_topo_phases_2023}, the local markers derived here are $\mathbb{Z}$ invariants that can identify multiple topological states per band gap beyond those predicted by a system's fractional filling anomaly (see Supplemental Material Sec.~SVIII \cite{SI}). 
Thus, further work is required to connect these two frameworks by accounting for the different phenomena each is sensitive to.} More immediately, our operator-based framework should be useful for the design of materials seeking to leverage crystalline topology to enhance interactions \textcolor{black}{by optimizing over the predicted measure of a state's topological robustness}. Furthermore, by providing a physically motivated derivation of our local markers, our work may aid future theoretical studies in finding local markers for other classes of topology, and in particular the topology seen in moir\'{e} systems.


\begin{acknowledgments}
A.C.\ and T.L.\ acknowledge support from the Laboratory Directed Research and Development program at Sandia National Laboratories. H.\ S.-B.\ acknowledges support from DFG under Grant No.\ SCHU 1358/8-1. T.L.\ acknowledges support from the National Science Foundation, grant DMS-2110398. A.C.\ acknowledges support from the U.S.\ Department of Energy, Office of Basic Energy Sciences, Division of Materials Sciences and Engineering. This work was performed, in part, at the Center for Integrated Nanotechnologies, an Office of Science User Facility operated for the U.S. Department of Energy (DOE) Office of Science. Sandia National Laboratories is a multimission laboratory managed and operated by National Technology \& Engineering Solutions of Sandia, LLC, a wholly owned subsidiary of Honeywell International, Inc., for the U.S.\ DOE's National Nuclear Security Administration under contract DE-NA-0003525. The views expressed in the article do not necessarily represent the views of the U.S.\ DOE or the United States Government.
\end{acknowledgments}


%

\end{document}


\title{Supplemental Material for Local Markers for Crystalline Topology}

\author{Alexander Cerjan}
\email[]{awcerja@sandia.gov}
\affiliation{Center for Integrated Nanotechnologies, Sandia National Laboratories, Albuquerque, New Mexico 87185, USA}

\author{Terry A.\ Loring}
\affiliation{Department of Mathematics and Statistics, University of New Mexico, Albuquerque, New Mexico 87131, USA}

\author{Hermann Schulz-Baldes}
\affiliation{FAU Erlangen-N{\" u}rnberg, Department Mathematik, Cauerstr.\ 11, D-91058 Erlangen, Germany}


\maketitle

\section{Properties of the symmetry-reduced spectral localizer \label{sec:smi}}
\label{sec-properties}

At the heart of the local topological invariant framework for crystalline symmetries developed in the main text are the spectral properties of the spectral localizer  (Eq.\ (1) in the main text) which in the present situation is built from two Hermitian operators $X$ and $H$ on a Hilbert space $\mathcal{H}$ (called the position operator and Hamiltonian, respectively): 
\begin{align*} 
    L_{(x,E)}(X,H)
& = (H-E{\bf 1})\otimes \sigma_x+\kappa(X-x{\bf 1}) \otimes \sigma_y
\\
& =       \begin{pmatrix}
        0 & (H-E{\bf 1})-i\kappa(X-x{\bf 1})\\
        (H-E{\bf 1})+i\kappa(X-x{\bf 1}) & 0
        \end{pmatrix}.
\end{align*}
Here $\sigma_x$, $\sigma_y$ and $\sigma_z$ are the three Pauli matrices, $x,E$ are real numbers and $\kappa>0$ is a tuning parameter that also guarantees consistent units. In general, one can build a spectral localizer from more Hermitian operators by using a larger-dimensional non-trivial Clifford representation. In order to compactify notations, we will simply drop the arguments $X$ and $H$ and write $L_{(x,E)}=L_{(x,E)}(X,H)$. In the present situation with two Hermitian operators $X$ and $H$ (or more generally an even number of them), the spectral localizer is odd with respect to the symmetry ${\bf 1} \otimes \sigma_z$:
$$
({\bf 1} \otimes \sigma_z) L_{(x,E)}({\bf 1} \otimes \sigma_z) = - L_{(x,E)}
\;,
$$
so that it is off-diagonal in the grading of the Pauli matrices. Moreover, by construction the spectral localizer $L_{(x,E)}$ is Hermitian. These two properties imply that its spectrum is real and symmetric around $0$. Note also that $L_{(x,E)}$ is invertible if and only if the off-diagonal entries are invertible. If invertibility is given,  the signature $\mbox{\rm sig}[ L_{(x,E)}]$ is well-defined. While $\mbox{\rm sig}[ L_{(x,E)}]$ is useful for identifying material topology in the case of an odd number of Hermitian operators (corresponding to a physical system with even dimensionality) \cite{LoringSchuBa_odd,LoringSchuBa_even}, in the present situation it vanishes due to the symmetry of the spectrum.

\vspace{.2cm}

It is possible, however, to extract important information from the spectral localizer if one has given another (unitary selfadjoint)  symmetry operator $\mathcal{S}$ on $\mathcal{H}$ satisfying, just as in the main text,
\begin{equation}
\mathcal{S}^2 = {\bf 1}, \qquad
\mathcal{S} = \mathcal{S}^\dagger, \qquad
\mathcal{S}H = H\mathcal{S}, \qquad
\mathcal{S}X = -X\mathcal{S}. 
\label{eq-SymRel}
\end{equation}
In the application in the main text, $\mathcal{S}$ is a crystalline symmetry, more precisely a reflection symmetry, but here we first consider general structural properties of the set-up given by \eqref{eq-SymRel}. The main point is the following conjugation with the unitary $\left(\begin{smallmatrix} \mathcal{S} & 0 \\ 0 & {\bf 1} \end{smallmatrix}\right)$ that can  readily be checked:
\begin{align}
    \begin{pmatrix}
        \mathcal{S} & 0\\
        0 & {\bf 1}
        \end{pmatrix}^\dagger L_{(x,E)}
    \begin{pmatrix}
        \mathcal{S} & 0\\
        0 & {\bf 1}
        \end{pmatrix} 
        &= \begin{pmatrix}
        0 & [((H-E{\bf 1})+i\kappa(X-x{\bf 1}))\mathcal{S}]^\dagger \\
        ((H-E{\bf 1})+i\kappa(X-x{\bf 1}))\mathcal{S} & 0
        \end{pmatrix}. \label{eq:Ltrans}
\end{align}
This unitary equivalence shows that the operator on the r.h.s.\ has the same spectrum as $L_{(x,E)}$. In general, the off-diagonal blocks on the r.h.s.\ of Eq.\ \eqref{eq:Ltrans} are non-Hermitian, but the commutation relations \eqref{eq-SymRel} imply that they {\it are} Hermitian at $x = 0$. Hence let us introduce the symmetry-reduced or $\mathcal{S}$-reduced spectral localizer $\tilde{L}^{\mathcal{S}}_{E}=\tilde{L}^{\mathcal{S}}_{E}(X,H)$  by 
$$
\tilde{L}^{\mathcal{S}}_{E}
\;=\;
((H-E{\bf 1}) + i\kappa X)\mathcal{S}
\;.
$$
The terminology reflects that $\tilde{L}^{\mathcal{S}}_{E}$ depends on the symmetry $\mathcal{S}$ and consists only of one block entry of the spectral localizer, and is hence reduced. However, as stated above, the relations \eqref{eq-SymRel} imply that it is nevertheless Hermitian and clearly one has 
\begin{equation}
\label{eq-LLtilde}   
\begin{pmatrix}
        \mathcal{S} & 0\\
        0 & {\bf 1}
        \end{pmatrix}^\dagger L_{(0,E)}
    \begin{pmatrix}
        \mathcal{S} & 0\\
        0 & {\bf 1}
        \end{pmatrix} 
= \begin{pmatrix}
0 & \tilde{L}^{\mathcal{S}}_{E} \\
\tilde{L}^{\mathcal{S}}_{E} & 0
\end{pmatrix}
.
\end{equation}
It is possible to rewrite the $\mathcal{S}$-reduced spectral localizer in a form that looks more like in other works. For that purpose, let us split the Hilbert space $\mathcal{H}$ into even and odd elements w.r.t.\ $\mathcal{S}$:
$$
\mathcal{H}=\mathcal{H}_+\oplus\mathcal{H}_-,
\qquad
\mathcal{H}_\pm=\big\{\boldsymbol{\uppsi}\in\mathcal{H}\,:\,\mathcal{S}\boldsymbol{\uppsi}=\pm\boldsymbol{\uppsi}\big\}.
$$
Now $H:\mathcal{H}_\pm\to \mathcal{H}_\pm$ leaves subspaces invariant and can thus be decomposed as $H=H_+\oplus H_-$. On the other hand, $X:\mathcal{H}_\pm\to \mathcal{H}_\mp$ is off-diagonal in the grading of $\mathcal{S}$ and thus let us use the notation $\tilde{X}=X|_{\mathcal{H}_+}:\mathcal{H}_+\to \mathcal{H}_-$ for the restriction of $X$. Then in the grading of $\mathcal{S}$:
\begin{equation}
\label{eq-GradingRep}
\tilde{L}^{\mathcal{S}}_{E}
=
\begin{pmatrix}
(H_+-E{\bf 1}_+) & -i \kappa \tilde{X}^\dagger \\  i \kappa \tilde{X} & -(H_--E{\bf 1}_-)  
\end{pmatrix},
\end{equation}
where ${\bf 1}_\pm$ is the identity on $\mathcal{H}_\pm$. This resembles the even spectral localizer in \cite{LoringSchuBa_even}, in particular in the general form of \cite{schulz2021spectral} where also not necessarily all components of the position operator enter into the construction. Another crucial property is that its square
\begin{equation}
\label{eq-L^2}
(\tilde{L}^{\mathcal{S}}_{E})^2
\;=\;
(H-E{\bf 1})^2 +i\kappa [X,H]\mathcal{S} + \kappa^2 X^2
\end{equation}
only involves the commutator $[X,H]$ and {\it not} a term like $XH+HX$. In applications to local operators, $[X,H]$ is uniformly bounded, but $XH+HX$ grows with the volume. Also note that $i\kappa [X,H]\mathcal{S}$ is a Hermitian operator which is odd w.r.t.\ $\mathcal{S}$, namely $\mathcal{S}^\dagger \big(i\kappa [X,H]\mathcal{S}\big)\mathcal{S}=-\big(i\kappa [X,H]\mathcal{S}\big)$. 

\vspace{.2cm}

For any Hermitian $A$, the spectrum of $\left(\begin{smallmatrix} 0  & A \\ A & 0 & \end{smallmatrix}\right)$  consists of the union of the spectrum of $A$ with its negative because, if $A\psi=\lambda\psi$, the two eigenvectors of $\left(\begin{smallmatrix} 0  & A \\ A & 0 & \end{smallmatrix}\right)$ are simply $\left(\begin{smallmatrix} \pm\psi \\ \psi \end{smallmatrix}\right)$. Thus
\begin{equation}
    \lambda \textrm{ or } -\lambda \textrm{ is an eigenvalue of } \tilde{L}^{\mathcal{S}}_{E}  \quad \Longleftrightarrow  \quad \pm \lambda \textrm{ are both eigenvalues of } {L}_{(0,E)}.
\label{eq-EigenRel}
\end{equation}
Another piece of spectral information of an arbitrary linear operator $A$ is its invertibility gap $\mu(A)\geq 0$ defined by
$$
\mu(A)=
\min\big\{\sqrt{|\lambda|}\,:\,\lambda\in\spec(A^\dagger A)\big\}=\sqrt{\mu(A^\dagger A)}.
$$
For matrices, $\mu(A)$ is also called the smallest singular value. If $A^\dagger=A$ is a Hermitian matrix (or more generally $A$ is normal, namely $AA^\dagger=A^\dagger A$), then $\mu(A)$ is the smallest of all absolute values of the eigenvalues of $A$. For the spectral localizer, let us introduce the notation 
$$
\mu_{(x,E)} = \mu_{(x,E)}(X,H) =\mu(L_{(x,E)}(X,H) ).
$$
Then \eqref{eq-EigenRel} implies
\begin{equation}
\label{eq-s_min}
\mu_{(0,E)}=
\mu(L_{(0,E)} )
=
\mu(\tilde{L}^{\mathcal{S}}_{E})
=
\mu\big((H-E{\bf 1}) + i\kappa X\big).
\end{equation}
In the following, let us suppose that the Hilbert space $\mathcal{H}$ is finite dimensional which is relevant for the numerical treatment in the main text.  Whenever $\mu_{(0,E)} >0$, one can define
$$
\zeta_{E}^{\mathcal{S}}=\tfrac{1}{2}\,\mbox{\rm sig}[\tilde{L}^{\mathcal{S}}_{E} ],
$$
which is an integer/half-integer if the dimension of the Hilbert space is even/odd. Changes of the half-signature are always integer and actually define the spectral flow through $0$, see Section~1 in \cite{doll2023spectral}: 
\begin{align}
\mbox{Sf}\big(E\in[E_0,E_1]\mapsto \tilde{L}^{\mathcal{S}}_{E}\big)
& =
\tfrac{1}{2}\big(\mbox{\rm sig}[\tilde{L}^{\mathcal{S}}_{E_1} ] -\mbox{\rm sig}[\tilde{L}^{\mathcal{S}}_{E_0} ]\big)
\label{eq-SF}
 =
\zeta_{E_1}^{\mathcal{S}}-\zeta_{E_0}^{\mathcal{S}}
,
\end{align}
provided that $\tilde{L}^{\mathcal{S}}_{E_1}$ and $\tilde{L}^{\mathcal{S}}_{E_0}$ are invertible.  Hence it is natural to consider the jump (or discontinuity or critical) points $E_c$ of the function $E\mapsto  \zeta_{E}^{\mathcal{S}}$ which are precisely the points at which there is spectral flow. These points $E_c$ can also be characterized as those points $E_c$ for which $\mbox{\rm Ker}(\tilde{L}^{\mathcal{S}}_{E_c})$ is non-trivial. The following proposition shows that in the vicinity of such jump points $E_c$ there must be an element in the spectrum of $H$.

\begin{prop}
\label{prop-SpectralIndicator}
Suppose $E_c\in \mathbb{R}$ is such that the dimension $m_c=\dim(\mbox{\rm Ker}(\tilde{L}^{\mathcal{S}}_{E_c}))$ is positive. Then $H$ has at least $m_c$ eigenvalues in $[E_c-\sqrt{\kappa\,\|[X,H]\|},E_c+\sqrt{\kappa\,\|[X,H]\|}]$, counted with their multiplicity.
\end{prop}

\begin{proof}  
By Courant's minmax principle it is sufficient to show that there is an $m_c$-dimensional subspace $\mathcal{E}_c$ of the Hilbert space such that  $\kappa \|[X,H]\| -(H-E_c{\bf 1})^2\geq 0 $ in the sense of positive operators. This will be verified for $\mathcal{E}_c=\mbox{\rm Ker}(\tilde{L}^{\mathcal{S}}_{E_c})$. Indeed, for the restrictions on this subspace one has
\begin{align*}
0
& = (\tilde{L}^{\mathcal{S}}_{E_c})^2\big|_{\mathcal{E}_c}
= \big((H-E_c{\bf 1})^2 +i\kappa [X,H]\mathcal{S} + \kappa^2 X^2
 \big)\big|_{\mathcal{E}_c}
\geq
\big((H-E_c{\bf 1})^2 -\kappa \|[X,H]\|
 \big)\big|_{\mathcal{E}_c}\;,
\end{align*}
just as required.
\end{proof}

\vspace{.2cm}

Proposition~\ref{prop-SpectralIndicator} does {\it not} exclude that there are further eigenvalues of $H$ in $[E_c-\sqrt{\kappa\,\|[X,H]\|},E_c+\sqrt{\kappa\,\|[X,H]\|}]$. For example, in Fig.~1 in the main there is a kernel of dimension $m
_c$ at $\omega = 0.5152(2\pi c/a)$, but besides an eigenvalue of $H$ at precisely this energy, there are certainly two further eigenvalues of $H$ in a small interval around it (at $\omega = 0.5121(2\pi c/a)$ and $\omega = 0.5125(2\pi c/a)$) that correspond to edge-localized states, see Sec.\ \ref{sec:edge_states}. Nevertheless, only the state at $\omega = 0.5152(2\pi c/a)$ is a corner state, as can readily be seen by looking at the plots in Fig.~\ref{fig:degen_states} of the field intensity $|E_z|^2$ of the two other eigenfunctions. 

\vspace{.2cm}

Let us outline a numerical procedure for how to isolate the topological boundary state corresponding to a jump of $E\mapsto \zeta^{\mathcal{S}}_E$ at $E_c$ by $+1$ and with $m_c=1$ which is the generic case (as higher degeneracies of the kernel are lifted by a generic symmetry-preserving perturbation). This may be of help in simulations at large volume when there are many eigenstates of $H$ close to $E_c$. 

\begin{enumerate}

\item First choose an interval $[E_c-\delta,E_c+\delta]$ for a small $\delta>0$ such that the interval contains an odd number of eigenvalues with normalized eigenstates $\boldsymbol{\uppsi}_1,\ldots,\boldsymbol{\uppsi}_{2n+1}$ of $H$, but no other $E$ with $\mu_{(0,E)}=0$. Order them such that the states with positive $\mathcal{S}$-parity are listed first. There should be $n+1$ states with positive $\mathcal{S}$-parity and $n$ with negative $\mathcal{S}$-parity (if this is not the case, one may have ``just missed'' one state and should slightly enlarge $\delta$; of course, $\delta$ should chosen sufficiently small so that $n$ is not too large). 

\item Then introduce the frame  $\Psi=(\boldsymbol{\uppsi}_1,\ldots,\boldsymbol{\uppsi}_{2n+1})$, which decomposes to $\Psi=(\Psi_+,\Psi_-)$ with $\mathcal{S}\Psi_\pm=\pm\Psi_\pm$. Hence $\Psi_+$ and $\Psi_-$ are of dimension  $n+1$ and $n$ respectively. One of the $n+1$ states in $\Psi_+$ is the topological boundary state. Introduce two diagonal matrices $E_\pm$ (of dimension $n+1$ and $n$) such that $H\Psi_\pm=\Psi_\pm E_\pm$. 

\item Then taking matrix element of \eqref{eq-GradingRep} leads to
\begin{equation}
\label{eq-PartialL}
\Psi^\dagger\tilde{L}^{\mathcal{S}}_{E}\Psi=
\begin{pmatrix}
E_+-E{\bf 1}_+ & -i\kappa  (\Psi_+)^\dagger X \Psi_-
\\
i\kappa  (\Psi_-)^\dagger X \Psi_+ & -(E_--E{\bf 1}_-)
\end{pmatrix}
,
\end{equation}
where the identities ${\bf 1}_\pm$ are of dimension $n+1$ and $n$. Now compute the $(n+1)\times n$ matrix $(\Psi_+)^\dagger X \Psi_-$ numerically. 

\item There should be one line which is considerably smaller than all others. Shift the corresponding eigenvector into the first row of $\Psi$, then $\boldsymbol{\uppsi}_1=\boldsymbol{\uppsi}_b$ is the desired topological boundary state with eigenvalue $E_b$, namely $H\boldsymbol{\uppsi}_b=E_b\boldsymbol{\uppsi}_b$. Supposing that $(\boldsymbol{\uppsi}_b)^\dagger X \Psi_-$ even vanishes (as it is really negligible compared to the rest), then the first row and first column of $\Psi^\dagger\tilde{L}^{\mathcal{S}}_{E}\Psi$ vanishes, except for the entry $(1,1)$ which is $E_b-E$, and hence this eigenvalue is effectively decoupled from all other states. Of the remaining entries of $(\Psi_+)^\dagger X \Psi_-$ there may be many of order $1$, and these then move the eigenvalues of $\tilde{L}^{\mathcal{S}}_{E}$ out of the kernel, provided that $\kappa$ is larger than the diagonal entries. This also explains why it is {\it not advantageous} to choose $\kappa$ too small, see Sec.\ \ref{sec:kappa} for further discussion.

\item Once the topological boundary state $\boldsymbol{\uppsi}_b$ is determined, one now has
$$
(\boldsymbol{\uppsi}_b)^\dagger \tilde{L}^{\mathcal{S}}_{E} \boldsymbol{\uppsi}_b \approx E_b-E +\mathcal{O}(\kappa),
$$
where the order $\kappa$ correction term results from the coupling with other states (also those not included in $\Psi$). In particular, the energy $E_c$ with non-trivial kernel of $\tilde{L}^{\mathcal{S}}_{E}$ satisfies
$$
E_b-E_c=\mathcal{O}(\kappa),
$$
which is considerably better than the estimate of Proposition~\ref{prop-SpectralIndicator} which only provides $E_b-E_c=\mathcal{O}(\sqrt{\kappa})$.

\item If the jump of $\zeta^{\mathcal{S}}_E$ at $E_c$ would be $-1$ (namely the spectral flow would be $-1)$, the same argument as above can be applied, except that $\boldsymbol{\uppsi}_b$ would be of negative $\mathcal{S}$-parity $s_b=-1$. As the sign in the lower right entry in \eqref{eq-PartialL} is different, this leads to a sign change. In both cases, one therefore has (see eq.~(10) in the main text)
\begin{equation}
(\boldsymbol{\uppsi}_b)^\dagger \tilde{L}^{\mathcal{S}}_{E} \boldsymbol{\uppsi}_b \approx s_b(E_c-E) +\mathcal{O}(\kappa).
\label{eq-Approx}
\end{equation}

\end{enumerate}

It is possible to add some mathematical rigor to the argument leading to \eqref{eq-Approx}, even in the case of Proposition~\ref{prop-SpectralIndicator} where $m_c$ is larger than $1$, but only for $\kappa$ very small. The reasoning will analyze the dependence of $\tilde{L}^{\mathcal{S}}_{E}$ on both $E$ and $\kappa$, and therefore it will rather be denoted by $\tilde{L}^{\mathcal{S}}_{\kappa,E}$, and similarly $\zeta^{\mathcal{S}}_{\kappa,E}$. The most important property of the matrix-valued  function $(\kappa,E)\in\mathbb{R}^2\mapsto \tilde{L}^{\mathcal{S}}_{\kappa,E}$ is that it is real analytic in both variables with values in the Hermitian matrices. Therefore Kato's analytic perturbation theory applies \cite{kato2013perturbation}, both to $\tilde{L}^{\mathcal{S}}_{\kappa,E}$ as well as $(\tilde{L}^{\mathcal{S}}_{\kappa,E})^2$ given in \eqref{eq-L^2}. In particular, for a given $\delta$ (that will be chosen to be smaller or of the order of $\kappa$), one can look at the spectral projection $P^\delta_{\kappa,E}=\chi(\tilde{L}^{\mathcal{S}}_{\kappa,E}\in(-\delta,\delta))$ of $\tilde{L}^{\mathcal{S}}_{\kappa,E}$ onto the interval $(-\delta,\delta)$. As long as its dimension is constant, is then known to depend analytically on $\kappa$ and $E$. 

\begin{prop}
\label{prop-SpectralIndicator2}
Suppose that for given $(\kappa,E_c)$ one has $m_c=\dim(\mbox{\rm Ker}(\tilde{L}^{\mathcal{S}}_{\kappa,E_c}))>0$.  Moreover, suppose that there is a $\delta>0$ such the $\kappa'\in[0,\kappa] \mapsto P^\delta_{\kappa',E_c}$ has constant dimension equal to $m_c$. Then the $\zeta$-jump at $E_c$ (or equivalently the spectral flow of the $\mathcal{S}$-reduced spectral localizer at $E_c$) is given by the sum of the $\mathcal{S}$-parities $s_j\in\{-1,1\}$, $j=1,\ldots, m_c$, of the $m_c$ eigenstates of $H$ with energies in the spectral interval $[E_c-\delta,E_c+\delta]$:
$$
\lim_{\epsilon\to 0}
\zeta_{\kappa,E_c+\epsilon}^{\mathcal{S}}-\zeta_{\kappa,E_c-\epsilon}^{\mathcal{S}}
=
-\mbox{\rm sig}\big[\mathcal{S}|_{\mbox{\rm\small Ran}(P^\delta_{0,E_c})}\big]
=
-\sum_{j=1}^{m_c}s_j,
$$
where $\mathcal{S}|_{\mbox{\rm\small Ran}(P^\delta_{0,E_c})}$ denotes the restriction of the quadratic form $\mathcal{S}$ to the range of the projection $P^\delta_{0,E_c}$.
\end{prop}

\begin{proof}
Due to \eqref{eq-SF}, the desired $\zeta$-jump is given given by the spectral flow of the $\mathcal{S}$-reduced spectral localizer at $E_c$. Moreover, it is well-known that this spectral flow can be computed via the signature of the so-called crossing form \cite{doll2023spectral}. Together, one gets for $\epsilon$ sufficiently small: 
$$
\zeta_{\kappa,E_c+\epsilon}^{\mathcal{S}}-\zeta_{\kappa,E_c-\epsilon}^{\mathcal{S}}
=
\mbox{Sf}\big(E\in[E_c-\epsilon,E_c+\epsilon]\mapsto \tilde{L}^{\mathcal{S}}_{\kappa,E}\big)
=
\mbox{\rm sig}\Big[\partial_E \tilde{L}^{\mathcal{S}}_{\kappa,E}\big|_{\mbox{\rm Ker}(\tilde{L}^{\mathcal{S}}_{\kappa,E_c})}\Big]
.
$$
The derivative can readily be read off from the definition of $\tilde{L}^{\mathcal{S}}_{\kappa,E}$:
$$
\zeta_{\kappa,E_c+\epsilon}^{\mathcal{S}}-\zeta_{\kappa,E_c-\epsilon}^{\mathcal{S}}
=
\mbox{\rm sig}\Big[-\mathcal{S}\big|_{\mbox{\rm Ker}(\tilde{L}^{\mathcal{S}}_{\kappa,E_c})}\Big]
=
-\,\mbox{\rm sig}\Big[\mathcal{S}\big|_{\mbox{\rm Ker}(\tilde{L}^{\mathcal{S}}_{\kappa,E_c})}\Big]
.
$$
Now one has ${\mbox{\rm Ker}(\tilde{L}^{\mathcal{S}}_{\kappa,E_c})}=
\mbox{\rm Ran}(P^\delta_{\kappa,E_c})$. By assumption and analytic perturbation theory, there is a constant such that $\|P^\delta_{\kappa,E_c}-P^\delta_{0,E_c}\|\leq C\kappa$. As the spectrum of $\mathcal{S}$ restricted to $\mbox{\rm Ran}(P^\delta_{0,E_c})$ is contained in $\{-1,1\}$, this implies that also $\mathcal{S}$ restricted to $\mbox{\rm Ran}(P^\delta_{\kappa',E_c})$ is invertible for all $\kappa'\in[0,\kappa]$. Consequently, the signature of these restrictions does not change and one concludes
$$
\zeta_{\kappa,E_c+\epsilon}^{\mathcal{S}}-\zeta_{\kappa,E_c-\epsilon}^{\mathcal{S}}
=
-\,\mbox{\rm sig}\Big[\mathcal{S}\big|_{\mbox{\rm Ran}(P_{\kappa,E_c}^\delta)}\Big]
=
-\,\mbox{\rm sig}\Big[\mathcal{S}\big|_{\mbox{\rm Ran}(P_{0,E_c}^\delta)}\Big]
.
$$
Finally, $\tilde{L}^{\mathcal{S}}_{0,E_c}=(H-E_c{\bf 1})\mathcal{S}$ and therefore $\mbox{\rm Ran}(P_{0,E_c}^\delta)$ is spanned by the eigenstates $\boldsymbol{\uppsi}_1,\ldots,\boldsymbol{\uppsi}_{m_c}$ of $H$ with energies lying in $[E_c-\delta,E_c+\delta]$. These states have $\mathcal{S}$-parties $s_j\in\{-1,1\}$ given by $\mathcal{S}\boldsymbol{\uppsi}_j=s_j\boldsymbol{\uppsi}_j$ for $j=1,\ldots,m_c$. Replacing this implies the claim.
\end{proof}

\section{Homotopy characterization for finite systems \label{sec:ii}}

This section proves the claim from the main text (see the discussion after Eq.~(4)) that $\zeta^{\mathcal{S}}_{E}$ classifies pairs $(X,H)$ of selfadjoint matrices satisfying $\mathcal{S}H=H\mathcal{S}$ and $\mathcal{S}X=-X\mathcal{S}$. The same argument applies to chiral Hamiltonians $JHJ=-H$ for some selfadjoint involution $J$ commuting with another selfadjoint matrix $X$ if one uses the chiral spectral localizer $\tilde{L}^J(X,H)=(\kappa X+iH)J$ as in \cite{Cheng2023}. Hence let us take a more general set-up. Suppose given a finite-dimensional Hilbert space equipped with a fixed selfadjoint unitary $\Pi$. For two Hermitian matrices $A$ and $B$ on this Hilbert space satisfying with $A\Pi=\Pi A$ and $B\Pi=-\Pi B$, let us use the $\Pi$-reduced spectral localizer
\begin{equation*}
\tilde{L}^{\Pi}(A,B)=\left(A+iB\right)\Pi.
\end{equation*}

\begin{lem} \label{lem:1D_same_localizer_index} 
Suppose given two matrices $A_{j}$ and $B_{j}$, $j=0,1$, satisfying
\begin{equation*}
A_{j}\Pi=\Pi A_{j}\quad\text{ and }\quad B_{j}\Pi=-\Pi B_{j},
\end{equation*}
and further suppose that both associated reduced spectral localizers are gapped:
\begin{equation*}
\min\left(\tilde{L}^{\Pi}(A_{j},B_{j})\right)\geq\delta>0.
\end{equation*}
There exists a continuous path of matrices $t\in[0,1]\mapsto (A_{t},B_{t})$ connecting $(A_0,B_0)$ to $(A_1,B_1)$ and satisfying
$$
A_{t}\Pi=\Pi A_{t}\quad\text{ and }\quad B_{t}\Pi=-\Pi B_{t}\quad\text{ and }\quad 
\mu\left(\tilde{L}^{\Pi}\left(A_{t},B_{t}\right)\right)\geq\delta,
$$
if, and only if, 
\begin{equation*}
\textup{sig}\left[\tilde{L}^{\Pi}(A_{0},B_{0})\right]=\textup{sig}\left[\tilde{L}^{\Pi}(A_{1},B_{1})\right].
\end{equation*}
\end{lem}

\begin{proof}
If such a path exists, then the path $\tilde{L}^{\Pi}(A_{t},B_{t})$
is a continuous path of invertible Hermitian matrices, along which the signature cannot change. The key to proving the converse is to use the formulas that take us
from $\tilde{L}^{\Pi}(A,B)$ back to  $A$ and $B$. These are
\begin{align*}
A  =\frac{1}{2}\left(\left(\tilde{L}^{\Pi}(A,B)\Pi\right)^{\dagger}+\tilde{L}^{\Pi}(A,B)\Pi\right),
\qquad
B  =\frac{i}{2}\left(\left(\tilde{L}^{\Pi}(A,B)\Pi\right)^{\dagger}-\tilde{L}^{\Pi}(A,B)\Pi\right).
\end{align*}
The matrices $\tilde{L}^{\Pi}(A_{0},B_{0})$ and $\tilde{L}^{\Pi}(A_{1},B_{1})$ are by assumption Hermitian, gapped and have the same signature. By a straightforward argument based on the spectral theorem (diagonalize with diagonal entries ordered to their size, linearly interpolate between the diagonal matrices, finally deform the unitaries using their matrix logarithm), there exists a continuous path $t\in[0,1]\mapsto L_t$ of invertible selfadjoints connecting $L_0=\tilde{L}^{\Pi}(A_{0},B_{0})$ to $L_1=\tilde{L}^{\Pi}(A_{1},B_{1})$. Then set
$$
A_{t}  =\frac{1}{2}\left(\left(L_{t}\Pi\right)^{\dagger}+L_{t}\Pi\right),\qquad
B_{t}  =\frac{i}{2}\left(\left(L_{t}\Pi\right)^{\dagger}-L_{t}\Pi\right).
$$
These are clearly continuous paths of Hermitian matrices. Also, as required
\begin{equation*}
2\Pi A_{t}=L_{t}+\Pi L_{t}\Pi=2A_{t}\Pi
,
\qquad
-2i\Pi B_{t}=L_{t}-\Pi L_{t}\Pi=2iB_{t}\Pi.
\end{equation*}
Finally
\begin{equation*}
\tilde{L}^{\Pi}(A_{t},B_{t})=\frac{1}{2}\left(\left(\Pi L_{t}+L_{t}\Pi\right)\Pi-\left(\Pi L_{t}-L_{t}\Pi\right)\Pi\right)=L_{t}
\end{equation*}
so the path is in the invertibles as required.
\end{proof}

\vspace{.2cm}

While the previous lemma provides the desired classification, the result can probably be strengthened. It should be possible that along the constructed path $t\mapsto (A_t,B_t)$ the locality of the operators can be preserved as physical intuition suggests. One way to approach the problem is to realize that the commutator $[A,B]$ of the two Hermitian matrices can be made arbitrarily small by choosing the tuning parameter $\kappa$ small (recall that $A=H$, $B=\kappa X$ and $\Pi=\mathcal{S}$ in the situation of the main text). Then Lin's theorem \cite{LinAlmostCommutingHermitian1994} states that there are nearby Hermitian operators $A'$ and $B'$ which commute. This allows to construct a short part connecting them which hence conserves locality, as pointed out in \cite{hastings2008topology}.  It is an interesting open question whether one can generalize Lin's theorem by guaranteeing that $A'$ and $B'$ satisfy the same symmetry relations $A'\Pi=\Pi A'$ and $B'\Pi=-\Pi B'$. We conjecture that this is true, given that similar statements have been proven for matrix symmetries corresponding to Altland-Zirnbauer symmetry classes AI and AII \cite{loring2013almost} as well as  classes C and D \cite{loring2016almost}.

\section{The scaling coefficient $\kappa$ \label{sec:kappa}}

The scaling coefficient $\kappa$ serves two roles in the spectral localizer: it guarantees consistent units between the two constituent operators, the position operators on one side and the Hamiltonian on the other, and it adjusts the effective weighting between these operators in the spectral localizer's spectrum. The tuning of $\kappa$ is crucial because the two limiting cases $\kappa\downarrow 0$ and $\kappa\uparrow\infty$ both yield objects of no interest. For $\kappa = 0$, the spectrum of $L_{(x,E)}$ is simply related to the spectrum of the underlying Hamiltonian, and thus the spectral localizer provides no additional information. Similarly, for $\kappa \uparrow \infty$, $L_{(x,E)}$ only reveals the distribution of the system's sites. Instead, for the spectral localizer to provide new and useful information about a system, $\kappa$ needs to be chosen to neither be too large nor too small. These comments directly transpose the symmetry-reduced spectral localizer $\tilde{L}_E^{\mathcal{S}}$ due to Eqs.\ \eqref{eq-LLtilde} and \eqref{eq-EigenRel}.

\vspace{.2cm}

The spectral localizer in several prior works \cite{LoringSchuBa_odd,LoringSchuBa_even,doll2021skew} is used to compute the bulk topological invariants for bulk tight-binding Hamiltonians $H_{\textrm{bulk}}$ that fall within one of the ten Altland-Zirnbauer classes \cite{schnyder2008,kitaev2009,ryu2010topological}. In this context, there are proven bounds on $\kappa$ that guarantee that the spectral localizer is gapped (i.e.\ $\mu \ne 0$) so that its signature is a well-defined and stable quantity. More precisely,  let $E$ be the energy where the system's topology is being evaluated, let $g$ be the size of the band gap of $H_{\textrm{bulk}}$ around that $E$, and $l$ the length from the center of the sample to the boundary of the finite system.  Then the bounds
\begin{align}
    \kappa & \le \frac{g^3}{12 \Vert H_{\textrm{bulk}} - E {\bf 1} \Vert \left(\sum_{j=1}^d\Vert [X_j,H_{\textrm{bulk}}] \Vert \right)}, \label{eq:bnd1} \\
    \kappa & \ge \frac{2g}{l}, \label{eq:bnd2}
\end{align}
guarantee that the gap of the spectral localizer $L_{(0,E)}$ around $0$ is at least $\frac{g}{2}$ (see Theorem 2 of Ref. \cite{LoringSchuBa_even} and Chapter~10 in \cite{doll2023spectral}). Due to Eq.\ \eqref{eq-LLtilde} this again transposes directly to the symmetry-reduced spectral localizer.  Let us stress that the result makes no other assumption on the Hamiltonian $H_{\textrm{bulk}}$ than the existence of a spectral gap. The Hamiltonian need not be periodic or otherwise be homogeneous in space (such as in a quasicrystal). In particular, the result also holds for Hamiltonians describing defects. 
However, if a Hamiltonian describing a defect has a boundary state at energy $E$, then the spectral gap $g$ vanishes at that energy and necessarily the gap of spectral localizer also vanishes. 
Hence, in this work, there is a shift of perspective on the spectral localizer (or rather its symmetry-reduced cousin) relative to these prior studies \cite{LoringSchuBa_odd,LoringSchuBa_even,doll2021skew}: it is not only used to detect (bulk) topological invariants, but also to localize topological bound states. In the following, we argue that for the purpose of detecting these boundary-localized states, an adequate choice for the size of $\kappa$ is nevertheless given by Eq.\ \eqref{eq:bnd2} if $g$ is chosen to be the bulk gap. 

\vspace{.2cm}

Before going on though, let us stress that the bounds on the gap of the spectral localizer are not available for unbounded Hamiltonians, such as for photonic systems where Maxwell's equation leads to an unbounded Hamiltonian. In particular, the hypothesis \eqref{eq:bnd1} is meaningless for unbounded Hamiltonians  where $\Vert H_{\textrm{bulk}} \Vert \rightarrow \infty$. To recover a useful bound this quantity needs to be made finite in some way. One possibility is to use the resolvent $\Vert H_{\textrm{bulk}} \Vert \rightarrow \Vert (H_{\textrm{bulk}} - E{\bf 1})^{-1} H_{\textrm{bulk}} \Vert$ (which would still need to have its units corrected in some appropriate manner). Another possibility is to project into the local-in-energy subspace $\Vert H_{\textrm{bulk}} \Vert \rightarrow \Vert \Psi^\dagger H_{\textrm{bulk}} \Psi \Vert $ where $\Psi$ is rectangular matrix whose $m$ columns are the eigenvectors of the $m$ eigenvalues of $H_{\textrm{bulk}}$ closest to the chosen $E$ \cite{dixon_classifying_2023}. However, at present, it is not known how either of these changes will alter the bounds, or if there is an entirely different, better approach. 

\vspace{.2cm}

Nevertheless, Eq.\ \eqref{eq:bnd2} is expressed in terms of quantities that can still be calculated for a system with an unbounded Hamiltonian. This connection \eqref{eq:bnd2} is used as a guiding principle for the choice of $\kappa$ with $g$ being the bulk gap (and {\it not} the energetic distance from $E$ to the boundary-localized state). For the system shown in Fig.\ 1 of the main text, the bulk gap is $g \approx 0.068 (2 \pi c)^2/a^2$ and the length of the system in $y$ is $l = 8.65 a$, yielding $\kappa \approx 0.016 (2 \pi c)^2/a^3$. As can be seen, this is very similar to the value of $\kappa = 0.010 (2 \pi c)^2/a^3$ used in those simulations.

\begin{figure*}[t]
    \centering
    \includegraphics{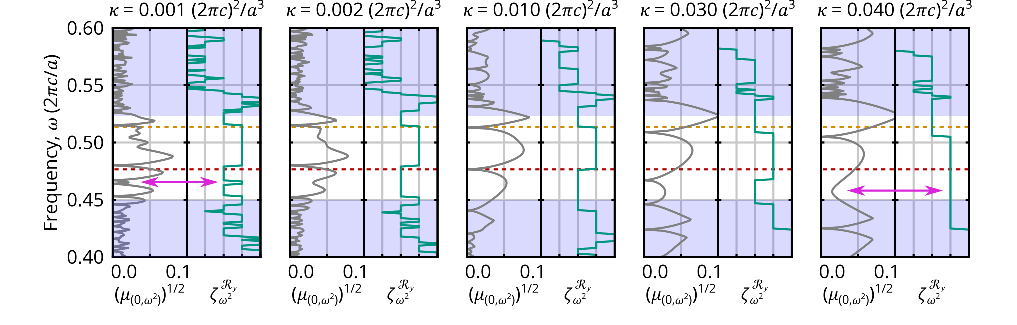}
    \caption{Local gap $(\mu_{(0,\omega^2)})^{(1/2)}$ in units of $2\pi c/a$ (left panels) and local index $\zeta_{\omega^2}^{\mathcal{R}_y}$ (right panels) calculated using the value of $\kappa$ indicated above each pair of panels for the system shown in Fig.\ 1 of the main text. The middle panel reproduces the data from Fig.\ 1(c) and (d). The shaded regions demarcate those spectral regions where states can exist in the crystalline bulk. For $\kappa = 0.001 (2\pi c)^2/a^3$, the magenta arrow indicates where the local index is changing for a pair of edge states. For $\kappa = 0.040 (2\pi c)^2/a^3$, the magenta arrow indicates where a corner state is no longer associated with a change in the local index.}
    \label{fig:kappa_check}
\end{figure*}

\vspace{.2cm}

By direct calculation, we confirm that $\kappa$ can be varied by over an order of magnitude while yielding quantitatively similar results. In Fig.\ \ref{fig:kappa_check}, we show the local gap and local index for a range of $\kappa$ surrounding the value predicted by Eq.\ \eqref{eq:bnd2}. The three central panels, $\kappa = [0.002,\; 0.010,\; 0.030] (2\pi c)^2/a^3$ all show nearly the same quantitative results, with two topological corner modes within the bulk band gap identified by a closing in the local gap where the local index also changes. The leftmost ($\kappa = 0.001 (2 \pi c)^2/a^3$) and rightmost ($\kappa = 0.040 (2 \pi c)^2/a^3$) panels both show how the theory begins to break down. For too small a $\kappa$, too many states are identified as topological in the system; in this case, two nearly degenerate edge-localized modes are identified as topological. Similarly, when $\kappa$ is too large, one of the in-gap topological corner states is effectively combined with an in-band state with the opposite topological charge.

\vspace{.2cm}

Heuristically, the behavior seen in Fig.\ \ref{fig:kappa_check} is demonstrating that the effect of $\kappa$ is to smooth out pairs of nearly degenerate modes with opposite symmetry, removing their presence from the local gap and local index. In this case, the in-gap nearly degenerate modes are made of edge-localized states that form an even and odd pair with respect to the reflection symmetry at the corner, whose frequencies are guaranteed to be similar, but distinct. However, if $\kappa$ is too large, some of the corner-localized modes are smoothed out with their in-band even/odd partner. This argument confirms the prediction that $\kappa$ should be both related to the gap size, which determines how far away such an in-band even/odd partner state is in energy, as well as the finite system's size, which sets the energy scale of the even-odd splitting between the edge-localized states (larger $l$ reduces the edge state splitting).

\section{Additional details on the perturbed corner heterostructure}

In Fig.\ 1 of the main text, the local gap is seen to indicate that the higher-frequency corner-localized state has more topological protection than the lower-frequency corner state, despite being closer-in-frequency to one of the bulk bands. Here, we provide additional simulation results beyond those shown in Fig.\ 2 of the main text to support the conclusions discussed in conjunction with that figure.

\vspace{.2cm}

\begin{figure*}[t]
    \centering
    \includegraphics{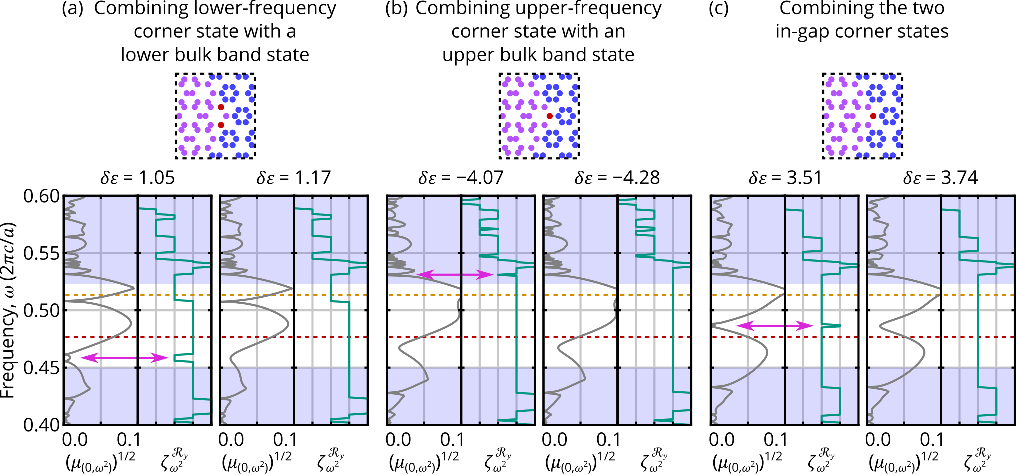}
    \caption{Studies of three different perturbations designed to annihilate one or both of the original in-gap topological corner states by combining the state with an opposite-symmetry partner. The top row shows zoomed-in diagrams of the perturbation's distribution (red) in the overall photonic corner heterostructure from Fig.\ 1 in the main text. The bottom left panel for each case shows the local gap $(\mu_{(0,\omega^2)})^{(1/2)}$ in units of $2\pi c/a$ and local index $\zeta_{\omega^2}^{\mathcal{R}_y}$ calculated using $\kappa = 0.01 (2\pi c)^2/a^3$ for a perturbation strength just below the strength at which the topological corner mode is annihilated. The bottom right panel shows the same quantities for a perturbation strength just above the strength at which the topological corner mode is annihilated. The magenta arrow indicates the topological feature that disappears as the perturbation strength is increased. The shaded regions demarcate those spectral regions where states can exist in the crystalline bulk. (a) The perturbation is chosen to combine: (a) the lower-frequency corner state with a state from the lower-frequency bulk band; (b) the higher-frequency corner state with a state from the higher-frequency bulk band; (c) the two in-gap corner-localized states.}
    \label{fig:perturb_sm}
\end{figure*}

In Fig.\ \ref{fig:perturb_sm} we consider the same two perturbation distributions as considered in the main text, tailored to affect either the lower-frequency corner state [Fig.\ \ref{fig:perturb_sm}(a)] or the higher-frequency corner state [Figs.\ \ref{fig:perturb_sm}(b),(c)], with the field profiles of these states in the unperturbed system shown in Fig.\ 1(e) of the main text. The total perturbation is then the perturbation strength $\delta \varepsilon\in\mathbb{R}$ multiplied by the perturbation distribution $u(\mathbf{x})$, i.e.\ $\delta \varepsilon_{\textrm{tot}}(\mathbf{x}) = \delta \varepsilon \, u(\mathbf{x})$. The two chosen perturbation distributions have similar overlaps with their respective modes, the higher-frequency corner state modal overlap is $\int E_{z,\textrm{hf}}^*(\mathbf{x}) u_{\textrm{hf}}(\mathbf{x}) E_{z,\textrm{hf}}(\mathbf{x}) d\mathbf{x} = 0.07 = \bar{u}_{\textrm{hf}}$, while the lower-frequency corner state modal overlap is $\int E_{z,\textrm{lf}}^*(\mathbf{x}) u_{\textrm{lf}}(\mathbf{x}) E_{z,\textrm{lf}}(\mathbf{x}) d\mathbf{x} = 0.12 = \bar{u}_{\textrm{lf}}$. Here, $u_{\textrm{lf}}(\mathbf{x})$ is shown in Fig.\ \ref{fig:perturb_sm}(a), and $u_{\textrm{hf}}(\mathbf{x})$ in Figs.\ \ref{fig:perturb_sm}(b) or (c), while the modal fields $E_z(\mathbf{x})$ of bound states of the unperturbed system are shown in Fig.\ 1(e) of the main text.

\vspace{.2cm}

To demonstrate that the topological phase transitions are happening at the perturbation strengths claimed in the main text, in Fig.\ \ref{fig:perturb_sm} we show closely spaced perturbation strengths on either side of the topological transition. In each case, the relevant topological that disappears as the perturbation strength is slightly increased is marked by the magenta arrow. In Fig.\ \ref{fig:perturb_sm}(c), we also show the possibility of attempting to merge the two in-gap corner states, as opposed to combining them with states from the bulk bands as is considered in Figs.\ \ref{fig:perturb_sm}(a) and (b). From Fig.\ 1(d), the local gap reveals that the topological protection of the lower-frequency corner state against merging with the lower-frequency bulk band is less than that of merging the two in-gap states together, which is again less than that of merging the higher-frequency corner state with the higher-frequency bulk band. This is because the maximum local gap attained over frequencies below the lower-frequency corner mode is less than the maximum local gap for frequencies in between the two corner states, which is again less than the local gap for frequencies greater than the higher-frequency corner state, 
\begin{equation*}
\max_{\omega < \omega_{\textrm{lf}}} \mu_{(0,\omega^2)} < \max_{\omega_{\textrm{lf}} < \omega <\omega_{\textrm{uf}}} \mu_{(0,\omega^2)} < \max_{\omega_{\textrm{uf}} < \omega} \mu_{(0,\omega^2)}.
\end{equation*}
This predicted hierarchy is confirmed by these the simulations considered in Fig.\ \ref{fig:perturb_sm}, 
\begin{equation*}
    |\delta \varepsilon_{\textrm{Fig.\ref{fig:perturb_sm}(a)}}| \bar{u}_{\textrm{lf}} < |\delta \varepsilon_{\textrm{Fig.\ref{fig:perturb_sm}(c)}}| \bar{u}_{\textrm{uf}} < |\delta \varepsilon_{\textrm{Fig.\ref{fig:perturb_sm}(b)}}| \bar{u}_{\textrm{uf}},
\end{equation*}
where the absolute value of each perturbation strength necessary to annihilate a specified corner-localized mode is scaled by the overlap with that corner-localized mode.

\vspace{.2cm}

In the numerical results discussed in this section and the main text, the perturbation was chosen to respect the system's reflection symmetry so that the symmetry-reduced spectral localizer remained Hermitian and the topological corner-localized states retained well-defined eigenvalues with respect to $\mathcal{R}_y$. However, it is likely possible to also consider the system's topological protection against perturbations that do not respect reflection symmetry. Of course, in this case, the corner-localized topological state would no longer be either even or odd with respect to the reflection symmetry, but it may be possible to quantitatively predict a perturbation strength below which the state must still exist. Such a prediction may be possible using known results for symmetry-destroying perturbations to systems with strong topology \cite{DollShcuba_Approx_symm_top_ins}, or by considering a non-Hermitian extension of the symmetry-reduced spectral localizer similar to what has been developed for Class A systems \cite{cerjan2023spectral}.

\section{Identifying topological edge-localized states}

\begin{figure*}[t]
    \centering
    \includegraphics{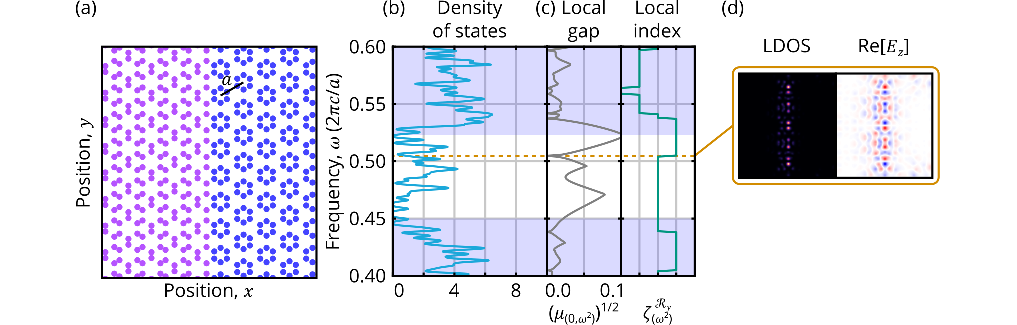}
    \caption{(a) Diagram of a 2D photonic structure with a straight edge between crystals formed from expanded (purple) and contracted (blue) hexagonal clusters that have the same properties as in Fig.\ 1 of the main text, bounded by a perfect electric conductor. (b) Density of states for the finite system in (a). (c) Local gap $(\mu_{(0,\omega^2)})^{(1/2)}$ in units of $2\pi c/a$ and local index $\zeta_{\omega^2}^{\mathcal{R}_y}$ calculated using $\kappa = 0.01(2\pi c)^2/a^3$. In (b,c) the shaded regions demarcate those spectral regions where states can exist in the crystalline bulk. (d) LDOS at the frequency of the local gap closing and real part of the $E_z$ field for the mode at $\omega = 0.504(2\pi c/a)$.}
    \label{fig:edge}
\end{figure*}

Re-configuring the corner photonic heterostructure considered in Fig.\ 1 of the main text so that the mirror symmetry axis intersects an edge rather than a corner, shows that the same local marker and local gap definitions can be used to identify topological edge-localized states centered at $y=0$, see Fig.\ \ref{fig:edge}. Indeed, as can be seen, it would be nearly impossible to uniquely identify the topological edge-localized state in this system from its DOS, due to all of the surrounding nearly degenerate states, Fig.\ \ref{fig:edge}(b). However, only a single state causes the local gap to close within the heterostructure's bulk band gap, where the topological marker also changes, Fig.\ \ref{fig:edge}(c), and this frequency corresponds to a state that is edge-localized, Fig.\ \ref{fig:edge}(d). Note that due to the system possessing perfect electric conductor boundaries (open boundaries in the standard language of condensed matter systems), one only expects a small number of edge states to be topological with respect to the reflection symmetry of the system's $y=0$ axis, as opposed to a number proportional to the length of the edge. To identify the remainder of the edge-localized states, one could instead impose periodic boundaries and define a set of reflection symmetries corresponding to the center and edges of each horizontal ribbon super-cell, similar to what is considered in Ref.\ \cite{velury_topological_2021}. Then, the set of all topological edge-localized states would be the union of all those states identified by each choice of reflection symmetry. Altogether, these simulations demonstrate that the local crystalline topological marker $\zeta_{E}^{\mathcal{S}}$ is applicable to both first-order and higher-order topology.

\section{Eigenstates of the symmetry-reduced  spectral localizer}

\begin{figure*}[t]
    \centering
    \includegraphics{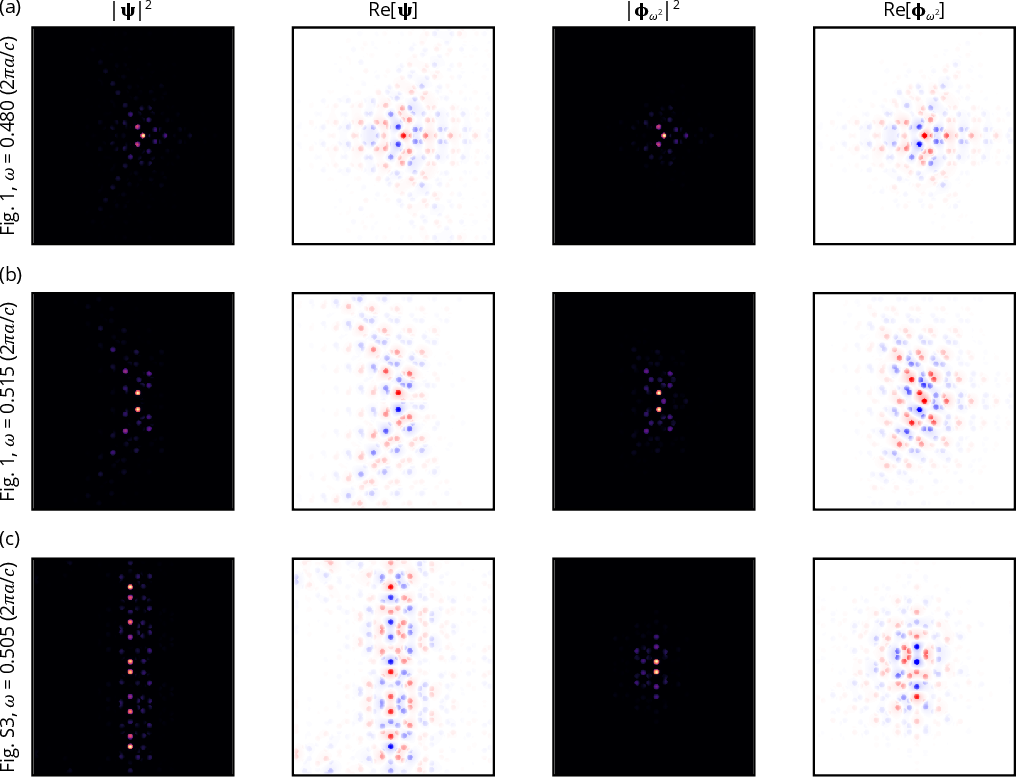}
    \caption{(a) For the system shown in the main text in Fig.\ 1, from left to right, the eigenstate intensity $|\boldsymbol{\uppsi}|^2$ of $H$ whose frequency is $\omega = 0.480(2\pi c/a)$, the corresponding real part of the eigenstate $\textrm{Re}[\boldsymbol{\uppsi}]$, the squared absolute value $|\boldsymbol{\upphi}_{\omega^2}|^2$ of the eigenvector $\boldsymbol{\upphi}_{\omega^2}$ of the $\mathcal{R}_y$-reduced spectral localizer whose corresponding eigenvalue is the closest to zero, and $\textrm{Re}[\boldsymbol{\upphi}_{\omega^2}]$.
    (b) Similar to (a), but for $\omega = 0.515(2\pi c/a)$ from that same system.
    (c) Similar to (a), but for $\omega = 0.505(2\pi c/a)$ for the system shown in Fig.\ \ref{fig:edge} in the main text. In all cases, $\kappa = 0.01(2\pi c)^2/a^3$.}
    \label{fig:loc_evecs}
\end{figure*}

In the main text after Eq.\ (10), it is claimed that the eigenstates of the $\mathcal{R}_y$-reduced spectral localizer at gap-closing are approximately given by eigenstates of the underlying Hamiltonian. In this section, we provide some numerical justification for this claim (as well as the discussion in Sec.\ \ref{sec:smi}), but also note that this is still a relatively crude approximation that somehow still results in a prediction, Eq.\ (12), that is numerically observed to be correct in Fig.\ 1 in the main text and Fig.\ \ref{fig:edge}. As a reminder, we are using $\boldsymbol{\upphi}_{\omega^2}$ as the eigenvector of the symmetry-reduced spectral localizer $\tilde{L}^{\mathcal{R}_y}_{\omega^2}$ that corresponds to its smallest eigenvalue at $\omega$. Similarly, the eigenvectors of the Hermitian Hamiltonian generated from Maxwell's equations are $\boldsymbol{\uppsi} = M^{-1/2} E_z$.

\vspace{.2cm}

As can be seen in Fig.\ \ref{fig:loc_evecs}(a),(b), $\boldsymbol{\upphi}_{\omega^2}$ is in reasonable agreement with the corresponding eigenstate of $H$, $\boldsymbol{\uppsi}_\textrm{b}$ for these corner-localized states. However, note that, as the symmetry-reduced spectral localizer does not commute with the spatial symmetry $\mathcal{S}$, $\boldsymbol{\upphi}_{\omega^2}$ is neither even nor odd with respect to this symmetry. For Fig.\ \ref{fig:loc_evecs}(c), we can see that even though the system's eigenstate is localized to the interface, but still extended along it, $\boldsymbol{\upphi}_{\omega^2}$ is localized to the reflection operator's center, i.e., the $y$-axis. Thus, the spectral localizer can identify topological states even if they are not perfectly localized.

\section{Examples of edge-localized states in the corner heterostructure that are trivial with respect to the local index \label{sec:edge_states}}

In the discussion about Fig.\ 1 in the main text, we claim that the other states seen in the shared bulk band gap of the photonic crystal heterostructure are edge-localized states and not corner-localized states. In Fig.\ \ref{fig:degen_states} we provide evidence for this claim. As can be seen, for frequencies corresponding to choices where the DOS is non-zero within the bulk band gap, but not frequencies where the local gap closes and the local index changes, the states are edge-localized, not corner-localized.

\begin{figure*}[h]
    \centering
    \includegraphics{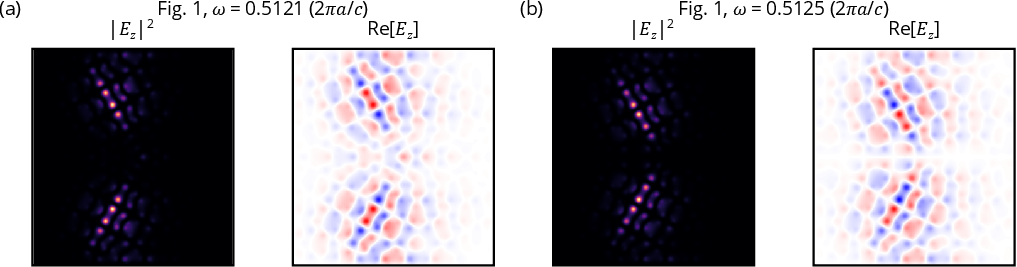}
    \caption{For the system shown in the main text in Fig.\ 1, $|E_z|^2$ and $\textrm{Re}[E_z]$ at (a)  $\omega = 0.5121(2\pi c/a)$ and at (b) $\omega = 0.5125(2\pi c/a)$.}
    \label{fig:degen_states}
\end{figure*}

\section{Application of the symmetry-reduced spectral localizer to a tight-binding model \label{sec:c6v}}

Here, we provide an example of the symmetry-reduced spectral localizer of a tight-binding model that identifies its crystalline topological states. In particular, we choose the ``breathing'' honeycomb lattice with $C_{6v}$ symmetry shown in Fig.\ \ref{fig:c6v}(a). This lattice is characterized by two coupling coefficients, the intra--unit cell coupling $t_{\textrm{in}}$, and the inter--unit cell coupling $t_{\textrm{out}}$. Here, we set the on-site energies to be zero, so the system also exhibits chiral symmetry. This lattice has been previously studied for its zero-energy states  protected by chiral symmetry \cite{noh_topological_2018}.

\begin{figure*}
    \centering
    \includegraphics{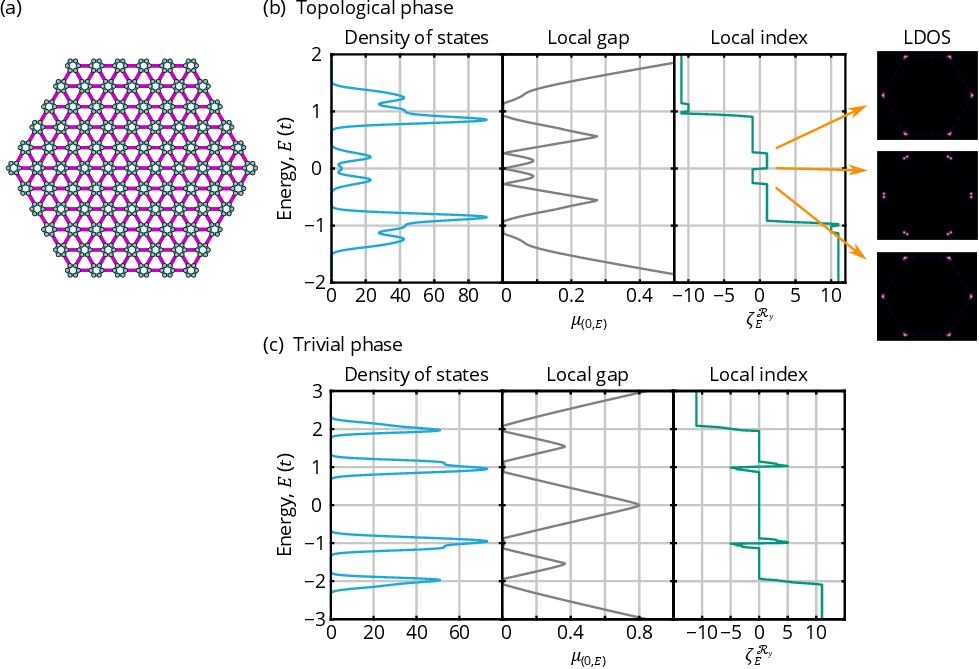}
    \caption{(a) Schematic of the finite breathing honeycomb lattice with $C_{6v}$ symmetry. The intra--unit cell couplings ($t_{\textrm{in}}$) are shown in cyan, and the inter--unit cell couplings ($t_{\textrm{out}}$) are shown in magenta, with lattice constant is $a$. (b) Density of states, local gap $\mu_{(0,E)}$, and local index $\zeta_{E}^{\mathcal{R}_y}$ for the breathing honeycomb lattice for $t_{\textrm{in}} = 0.2t$ and $t_{\textrm{out}} = t$. The local density of states (LDOS) is shown for the three in-gap energies where the topological index is seen to change. (c) Similar to (b), except for $t_{\textrm{in}} = t$ and $t_{\textrm{out}} = 0.2t$. For (b),(c), $\kappa = 0.1 (t/a)$.}
    \label{fig:c6v}
\end{figure*}

\vspace{.2cm}

Instead, here we analyze the breathing honeycomb lattice using the symmetry-reduced spectral localizer. Again, we choose $\mathcal{S} = \mathcal{R}_y$, the reflection symmetry about the system's $y=0$ axis, yielding the local index $\zeta_{E}^{\mathcal{R}_y}$ and local gap $\mu_{(0,E)}$ as defined in the main text. When, $t_{\textrm{out}} > t_{\textrm{in}}$ [Fig.\ \ref{fig:c6v}(b)], the system exhibits three in-gap index switches, showing that the system's topology is changing at these energies. The local density of states (LDOS) reveals that all three of these switches correspond to corner-localized states. In contrast, when $t_{\textrm{in}} > t_{\textrm{out}}$ [Fig.\ \ref{fig:c6v}(c)], there are no in-gap topological index changes, only changes that occur within the spectral extent of the system's bulk bands.

\vspace{.2cm}

Note that the system's local index in the topological phase only changes by $2$ at the energy of each set of corner states, not $6$, despite there being six corners. This is because only two of the corner states are topological with respect to $\mathcal{R}_y$ --- the other four corner states are topological with respect to the two reflection axes that bisect those corners. Similarly, the index is changing by $2$, rather than $1$, because each reflection axis is bisecting two corners on opposite sides of the lattice. In contrast, in Fig.\ 1 of the main text, the system only has a single corner along the reflection symmetry's axis.


%